\documentclass[10pt, conference, compsocconf]{IEEEtran}
\ifCLASSINFOpdf
\else
\fi

\usepackage{amsmath}
\usepackage{amssymb}
\usepackage{amsthm}
\usepackage[usenames]{color}
\usepackage{graphicx}
\usepackage{float}
\usepackage{pdfsync}
\usepackage{hyperref}
\usepackage{algorithmic}
\usepackage[ruled,vlined]{algorithm2e} 
\usepackage{enumerate}
\usepackage{fancyhdr}
\usepackage{paralist}

\renewcommand{\geq}{\geqslant}
\renewcommand{\leq}{\leqslant}

\newcommand{\equals}{\stackrel{\mathrm{def}}{=}}

\newtheorem{theorem}{Theorem}
\newtheorem{property}{Property}

\newtheorem{definition}{Definition}

\fancypagestyle{temp}{\cfoot{\thepage\hspace{1pt}}}
\begin{document}
%
\title{\huge Energy Minimization for Parallel Real-Time Systems with Malleable Jobs and Homogeneous Frequencies}

\author{\IEEEauthorblockN{Nathan Fisher}
\IEEEauthorblockA{Wayne State University\\
fishern@wayne.edu}
\and
\IEEEauthorblockN{Jo\"el Goossens}
\IEEEauthorblockA{Universit\'e Libre de Bruxelles\\
joel.goossens@ulb.ac.be}
\and
\IEEEauthorblockN{Pradeep M. Hettiarachchi}
\IEEEauthorblockA{Wayne State University\\
pradeepmh@wayne.edu}
\and
\IEEEauthorblockN{Antonio Paolillo}
\IEEEauthorblockA{Universit\'e Libre de Bruxelles\\
antonio.paolillo@ulb.ac.be}}


%


\maketitle

\begin{abstract}
In this work, we investigate the potential utility of parallelization for meeting real-time constraints and minimizing energy. We consider malleable Gang scheduling of implicit-deadline sporadic tasks upon multiprocessors. We first show the non-necessity of \emph{dynamic} voltage/frequency regarding optimality of our scheduling problem. We adapt the canonical schedule for DVFS multiprocessor platforms and propose a polynomial-time optimal processor/frequency-selection algorithm.  We evaluate the performance of our algorithm via simulations using parameters obtained from a hardware testbed implementation.  Our algorithm has up to a 60 watt decrease in power consumption over the optimal non-parallel approach.
\end{abstract}


%
\IEEEpeerreviewmaketitle

\section{Introduction}\label{sec:intro}
Power-aware computing is at the forefront of embedded systems research due to market demands for increased battery life in portable devices and decreasing the carbon footprint of embedded systems in general.  The drive to reduce system power consumption has led embedded system designers to increasingly utilize multicore processing architectures.  An oft-repeated benefit of multicore platforms over computationally-equivalent single-core platforms is increased energy efficiency and thermal dissipation~\cite{NVIDIA_Ref}.  For these power benefits to be fully realized, a computer system must possess the ability to parallelize its computational workload across the multiple processing cores.  However, parallel computation often comes at a cost of increasing the total, overall computation that the system must perform due to communication and synchronization overhead of the cooperating parallel processes.  In this paper, we explore the trade-off between parallelization of real-time applications and energy consumption.

\section{Related Work}\label{sec:related}

There are two main models of parallel tasks (i.e., tasks that may use several processors \emph{simultaneously}): the \emph{Gang}~\cite{CCG06b,JRWRTC11,Goossens2010Gang-FTP-schedu,kato2009gang} and the \emph{Thread} model~\cite{CouLupGoo12,5702236, parallel_task-Gill-2011}. With the Gang model, all parallel instances of a same task start and stop using the processors synchronously. On the other hand, with the Thread model, there is no such constraint. Hence, once a thread has been released, it can be executed on the processing platform independently of the execution of the other threads.

Very little research has addressed both real-time parallelization and power-consumption issues~\cite{kong2011energy,cho2007corollaries}.  Furthermore, some basic fundamental questions on the potential utility of parallelization for meeting real-time constraints and minimizing energy have not been addressed at all in the literature.

\section{Models}\label{sec:models}

\subsection{Parallel Job Model}\label{subsec:job-model}

In real-time systems, a job $J_\ell$ is characterized by its \emph{arrival time} $A_\ell$, \emph{execution time} $E_\ell$, and relative deadline $D_\ell$.  The interpretation of these parameters is that the system must schedule $E_\ell$ units of execution on the processing platform in the interval $[A_\ell, A_\ell + D_\ell)$.  Traditionally, most real-time systems research has assumed that the execution of $J_\ell$ must occur sequentially (i.e., $J_\ell$ may not execute concurrently with itself on two --- or more --- different processors).  However, in this paper, we deal with jobs which may be executed on different processors at the very same instant, in which case we say that \emph{job parallelism} is allowed. Various kind of task models exist; Goossens et al.~\cite{Goossens2010Gang-FTP-schedu} adapted parallel terminology~\cite{Buyya99} to recurrent (real-time) tasks as follows.

\begin{definition}[Rigid, Moldable and Malleable Job]
A \emph{job} is said to be (i) \emph{rigid} if the number of processors assigned to this job is specified externally to the scheduler a priori, and does not change throughout its execution;
(ii) \emph{moldable} if the number of processors assigned to this job is determined by the scheduler, and does not change throughout its execution;
(iii) \emph{malleable} if the number of processors assigned to this job can be changed by the scheduler during the job's execution.
\end{definition}

As a starting point for investigating the tradeoff between energy consumption and parallelism in real-time systems, we will work with the malleable job model in this paper.  Schedulability analysis is more complicated for the rigid and moldable job models, and we defer study of these models for future research.

\subsection{Parallel Task Model}\label{subsec:task-model}

In real-time systems, jobs are generated by (recurring) tasks.  One general and popular real-time task model is the \emph{sporadic task model}~\cite{thesisMok} where each sporadic task is characterized by its \emph{worst-case execution time} $e_i$, \emph{task relative deadline} $d_i$, and \emph{minimum inter-arrival time} $p_i$ (also called the task's \emph{period}).  A task $\tau_i$ can generate a (potentially) infinite sequence of jobs $J_1, J_2, \ldots$ such that: 1) $J_1$ may arrive at any time after system start time; 2) successive jobs must be separated by at least $p_i$ time units (i.e., $A_{\ell + 1} \geq A_\ell + p_i$); 3) each job has an execution requirement no larger than the task's worst-case execution time (i.e., $E_\ell \leq e_i$); and 4) each job's relative deadline is equal to the the task relative deadline (i.e., $D_\ell = d_i$).  A useful metric of a task's computational requirement upon the system is \emph{utilization} denoted by $u_i$ and computed by $e_i/p_i$.  A collection of sporadic tasks $\tau \equals \{\tau_1, \tau_2, \ldots, \tau_n\}$ is called a \emph{sporadic task system}.  In this paper, we assume a common subclass of sporadic task systems called \emph{implicit-deadline sporadic task systems} where each $\tau_i \in \tau$ must have relative deadline equal to its period (i.e., $d_i = p_i$).

At the task level, the literature distinguishes between at least two kinds of parallelism:

\begin{compactitem}
\item \emph{Multithread.} Each task is sequence of phases, each phase is composed of several threads, each thread requires a single processor for execution and \emph{can} be scheduled simultaneously~\cite{Nelissen2012Techniques-Opti}. A particular case is the \emph{Fork-Join} task model where task begins as a single master thread that executes sequentially until it encounters the first fork construct, where it splits into multiple parallel threads
which execute the parallelizable part of the computation~\cite{5702236} and so on.
\item \emph{Gang.} Each task corresponds to $e \times k$ rectangle where $e$ is the execution time requirement and $k$ the number of required processors with the restriction the $k$ processors execute task in unison~\cite{kato2009gang}.
\end{compactitem}

\noindent In this paper, we assume malleable Gang task scheduling.

Due to the overhead of communication and synchronization required in parallel processing, there are fundamental limitations on the speedup obtainable by any real-time job.  Assuming that a job $J_{\ell}$ generated by task $\tau_i$  is assigned to $k_\ell$ processors for parallel execution over some $t$-length interval,  the speedup factor obtainable is $\gamma_{i,k_{\ell}}$.  The interpretation of this parameter is that over this $t$-length interval $J_\ell$ will complete $\gamma_{i,k_\ell}\cdot t$ units of execution.   We let $\Gamma_i = (\gamma_{i,0}, \gamma_{i,1}, \ldots, \gamma_{i, m}, \gamma_{i,m+1})$ denote the multiprocessor speedup vector for jobs of task $\tau_i$ (assuming $m$ identical processing cores).  The variables $\gamma_{i,0}$ and $\gamma_{i,m+1}$ are sentinel values used to simplify the algorithm of Section~\ref{sec:algo}; the values of $\gamma_{i,0}$ and $\gamma_{i,m+1}$ are $0$ and $\infty$ respectively.  Throughout the rest of the paper, we will characterize a parallel sporadic task $\tau_i$ by $(e_i, p_i, \Gamma_i)$.

We consider the following two restrictions on the multiprocessor speedup vector:
\begin{compactitem}
\item \emph{Sub-linear speedup ratio~\cite{kong2011energy}}: $1 < \frac{\gamma_{i,j'}}{\gamma_{i,j}} < \frac{j'}{j}$ where $0<j<j'\leq m$.
\item \emph{Work-limited parallelism~\cite{CCG06b}}:
$\gamma_{i,(j'+1)}-\gamma_{i,j'}\leq\gamma_{i,(j+1)}-\gamma_{i,j}$ where $0\leq j<j' < m$.
\end{compactitem}

\noindent  The sub-linear speedup ratio restriction represents the fact that no task can truly achieve an ideal or better than ideal speedup due to the overhead in parallelization.  It also requires that the speedup factor strictly increases with the number of processors.  The work-limited parallelism restriction ensures that the overhead only increases as more processors are used by the job.  These restrictions place realistic bounds on the types of speedups observable by parallel applications.

\subsection{Power/Processor Model}\label{subsec:power-model}

We assume that the parallel sporadic task system $\tau$ executes upon a multiprocessor platform with $m$ identical processing cores.  The processing platform is enabled with both dynamic power management (DPM) and dynamic voltage and frequency scaling (DVFS) capabilities.  With respect to DPM capabilities, we assume the the processing platform has the ability to turn off any number of cores between 0 and $m-1$.  For DVFS capabilities, in this work, we assume that there is a system-wide homogeneous frequency $f >0$ which indicates the frequency at which all cores are executing at any given moment.  The power function $P(f, k)$ indicates the power dissipation rate of the processing platform when executing with $k$ active cores at a frequency of $f$.  We only assume that $P(f,k)$ is a non-decreasing, convex function.  While we consider the setting where the system may dynamically change frequency without penalty, we consider that there is significant overhead to turning a core on or off.  Therefore, in this paper, we will only consider core speed/activation assignment schemes where the number of active cores is decided prior to system runtime and does not change dynamically.

The interpretation of the frequency is that if $\tau_i$ is executing job $J_\ell$ on $k_\ell$ processors at frequency $f$ over a $t$-length interval then it will have executed $t\cdot \gamma_{i,k_\ell} \cdot f$ units of computation.  The total energy consumed
by executing $k$ cores over the $t$-length at frequency $f$ is $t\cdot P(f,k)$.

\subsection{Scheduling Algorithm}\label{subsec:scheduler}

In this paper, we use a scheduling algorithm originally developed for non-power-aware parallel real-time systems called the \emph{canonical parallel schedule}~\cite{CCG06b}.  The canonical scheduling approach is optimal for implicit-deadline sporadic real-time tasks with work-limited parallelism and sub-linear speedup ratio upon an identical multiprocessor platform (i.e., each processor has identical processing capabilities and speed).  In this paper, we consider also an identical multiprocessor platform, but permit both the number of active processors
and homogeneous frequency $f$ for all active processors to be chosen prior to system runtime.  In this subsection, we briefly define the canonical scheduling approach with respect to our power-aware setting.

Assuming the processor frequencies are identical and a fixed value $f$, it can be noticed that a task $\tau_{i}$ requires more than $k$ processors simultaneously if $u_{i}>\gamma_{i,k}\cdot f$; for the unitary frequency, we denote by $k_{i}$ the largest such $k$ (meaning that $k_{i}$ is the smallest number of processor[s] such that the task $\tau_{i}$ is schedulable on $k_{i}+1$ processors at frequency $f = 1$):
\begin{equation}\label{eq:smalNbrProc-orig}
  k_{i} \equals
\begin{cases}
0 & \text{if $u_{i}\leq \gamma_{i,1}$}\\
\max_{k=1}^{m} \{k \mid \gamma_{i,k} < u_i\} &
\text{otherwise.}
\end{cases}
\end{equation}

For example, let us consider the task system $\tau=\{\tau_1,
\tau_2 \}$ to be scheduled on three processors with $f=1$. We have $\tau_1=(6,4,\Gamma_1)$ with $\Gamma_1=(1.0, 1.5, 2.0)$ and $\tau_2=(3,4,\Gamma_2)$ with $\Gamma_2=(1.0, 1.2, 1.3)$. Notice that the system is infeasible at this frequency if job parallelism is not allowed since $\tau_1$ will never meet its deadline unless it is scheduled on at least two processors (i.e., $k_{1}=1$). There is a feasible schedule if the task $\tau_1$ is
scheduled on two processors and $\tau_2$ on a third one (i.e., $k_{2}=0$).

\paragraph*{The canonical schedule} That scheduler assigns $k_{i}$ processor(s) permanently to $\tau_{i}$ and an additional processor sporadically (see~\cite{CCG06b} for details). In this work we will extend that technique for dynamic voltage and frequency scaling (DVFS) and dynamic power management (DPM) capabilities.

\section{Non-Necessity of DVFS for Malleable Jobs}\label{sec:non-necessity}

\begin{property}\label{prop:dvfs_notnecessary}
In a multiprocessor system with global homogeneous frequency in a continuous range, choosing dynamically the frequency is not necessary for optimality in terms of consumed energy.
\end{property}

\begin{proof}
\cite{Ishihara1998Voltage-schedul} presented similar result, here we prove the property for our framework.
Although we have a proof of this property for any convex form of $P(v)$ ($v$ is the voltage chosen, directly linked to the resulting frequency of the system), for space limitation in the following, we will consider that $P(v) \propto v^3$.
Assuming we have a schedule at the constant speed/voltage $v$ on the (multiprocessor) platform we will show that any dynamic frequency schedule (which schedules the same amount of work) consumes not less energy. First notice that from any dynamic frequency schedule we can obtain a constant frequency schedule (which schedules the same amount of work) by applying, sequentially, the following transformation: given a dynamic frequency schedule in the interval $[a,b]$ which works at voltage $v_{1}$ in $[a,\ell)$ and at voltage $v_{2}$ in $[\ell,b]$ we can define the constant voltage such that at that speed/voltage the amount of work is identical.

Without loss of generality we will consider the constant voltage schedule the interval $[0,1]$ working at voltage $v$ and the dynamic schedule working at voltage $v+\Delta$ in $[0,\ell)$ and at the voltage $v-\Delta'$ in $[\ell,1]$.

Since the transformation must preserve the amount of work completed we must have:
\begin{align}
& v = \ell(v+\Delta) + (1-\ell)(v-\Delta') \notag \\
\Leftrightarrow \; & \Delta' \equals \frac{\ell  \Delta}{1-\ell} \label{equation:delta}
\end{align}
since the extra work in $[0,\ell)$ (i.e., $\Delta \ell$) must be equal to the spare work in $[\ell,1]$ (i.e., $\Delta' (1-\ell)$).

Now we will compare the \emph{relative} energy consumed by both the schedules, i.e., we will show that

\begin{equation}\label{equation:power}
\ell (v+\Delta)^3 + (1-\ell)(v-\frac{\ell  \Delta}{1-\ell})
^3 \geq v^3
\end{equation}

We know that $\ell (v+\Delta)^3 = \ell (v^3+3v^{2}\Delta+3v\Delta^{2}+\Delta^{3})$ and $(1-\ell)(v-\frac{\ell  \Delta}{1-\ell})^3 = (1-\ell)(v^{3}-3v^{2}\frac{\ell  \Delta}{1-\ell} + 3v\frac{\ell^{2}  \Delta^{2}}{(1-\ell)^{2}} - \frac{\ell^3  \Delta^3}{(1-\ell)^3})$.

(\ref{equation:power}) is equivalent to (by subtracting $v^{3}$ on the both sides)

$$\Delta \left[3\ell v \Delta + \ell \Delta^{2} + 3v\frac{\ell^{2}  \Delta}{(1-\ell)} - \frac{\ell^3  \Delta^2}{(1-\ell)^2}\right]\geq0$$

Or equivalently (dividing by $\ell\Delta$):

$$3\Delta v + \Delta^{2} + 3v\frac{\ell  \Delta}{(1-\ell)} - \frac{\ell^2  \Delta^2}{(1-\ell)^2}\geq0$$

$\Leftarrow$ ($v - \Delta' > 0$ and, by~(\ref{equation:delta}))
$$3\Delta v + \Delta^{2} + 3 \frac{\ell\Delta}{1-\ell} \frac{\ell  \Delta}{(1-\ell)} - \frac{\ell^2  \Delta^2}{(1-\ell)^2}\geq0$$

$\Leftarrow$
$$3\Delta v + \Delta^{2} + 2 \frac{\ell^2  \Delta^2}{(1-\ell)^2}\geq0$$

\noindent which always holds because $\Delta > 0$ and $v > 0$.
\end{proof}

\section{Optimal Processor/Frequency-Selection Algorithm}\label{sec:algo}
Property~\ref{prop:dvfs_notnecessary} implies, for homogeneous frequency upon the different processing cores, that for each DVFS scheduling, it exists a constant frequency scheduling which consumes no more energy. Thus, the frequency that minimizes consumed energy can be computed prior the execution of the system. So, in the following, we will design an offline algorithm to find this optimal minimal frequency. This parameter will allow us to use the canonical schedule~\cite{CCG06b} to find a scheduling of the system. First, we will present the feasibility criteria adapted to variable homogeneous frequency. After that we will use this criteria to determine constraints on the frequency for the system to be feasible on a fixed number of processors. After that, we will present an algorithm which uses those constraints to compute the exact optimal frequency for the system to be feasible. Finally, we will prove the correctness of this algorithm.

In the following we denote by $f$ the frequency of our multiprocessor platform. Notice that we made the hypothesis that time is continuous (as in~\cite{CCG06b}). More specifically, we can also choose the frequency in the positive continuous range ($f \in \mathbb{R}^+_0$).

\subsection{Background}\label{subsec:algo-bkgrd}
Notice that a task $\tau_{i}$ requires more than $k$
processors simultaneously if $u_{i}>\gamma_{i,k} \times f$; we denote
by $k_{i}(f)$ the largest such $k$ (meaning that $k_{i}(f)$ is the
smallest number of processor[s] such that the task
$\tau_{i}$ is schedulable on $k_{i}(f)+1$ processors):

\begin{equation}\label{eq:smalNbrProc}
  k_{i}(f) \equals
\begin{cases}
0, & \text{if $u_{i}\le \gamma_{i,1} \times f$}\\
\max_{k=1}^{m} \{k \mid \gamma_{i,k}\times f < u_i\}, &
\text{otherwise.}
\end{cases}
\end{equation}

This definition extends the one of $k_i$, Eq.~(\ref{eq:smalNbrProc-orig}). Notice that we have $k_{i} = k_{i}(1)$.
For a given number of processors $\kappa \in \{0, \ldots, m-1, m\}$, we wish to determine the range of frequencies $[f_1, f_2)$ such that $k_i(f) = \kappa$ for all $f \in [f_1, f_2)$.  We denote as the inverse function

\begin{equation}\label{eqn:inverse-smalNbrProc}
k_i^{-1}(\kappa) \equals
\begin{cases}
\{f~|~ \frac{u_i}{\gamma_{i,\kappa+1}} \leq f < \frac{u_i}{\gamma_{i,\kappa}}\} & \text{if $0 < \kappa \leq m$}\\
[\frac{u_i}{\gamma_{i,1}}, \infty) & \text{otherwise.}
\end{cases}
\end{equation}

\noindent We denote the left endpoint (resp., right endpoint) of $k^{-1}_i(\kappa)$ as $k^{-1}_i(\kappa).f_1$ (resp., $k^{-1}_i(\kappa).f_2$).

\subsection{Feasibility criteria with variable homogeneous frequency}\label{subs:feasibility}

We will now present a necessary and sufficient condition for the feasibility of a task system $\tau$ on $m$ identical processors at frequency $f > 0$.

\begin{theorem}\label{theo:tauprime}
A sporadic task system $\tau \equals \lbrace \tau_1, \tau_2, \ldots, \tau_n \rbrace$ is feasible on an identical platform with $m$ processing cores at frequency $f > 0$ if and only if the task system $\tau^\prime \equals \lbrace \tau^\prime_1, \tau^\prime_2, \ldots, \tau^\prime_n \rbrace$ is feasible on the same system with $m$ processing cores at frequency $1$.

$\tau^\prime$ is defined as follow:
\begin{align*}
&\forall \: 1 \leq i \leq n : \tau^\prime_i = (e_i, p_i,\Gamma^\prime_i) \\
&\Gamma^\prime_i  = (\gamma^\prime_{i,0}, \gamma^\prime_{i,1}, \ldots, \gamma^\prime_{i, m}, \gamma^\prime_{i,m+1}) \\
&\forall \: 0 \leq k \leq m+1 : \; \gamma^\prime_{i, k} \equals \gamma_{i, k} \times f.
\end{align*}

\end{theorem}

\begin{proof}
First of all, it is easy to see that $\tau$ respects \emph{sub-linear speedup ratio} and \emph{work-limited parallelism} if and only if $\tau^\prime$ respects them also.

We know that if $\tau_i$ is executing a job on $k$ processors at frequency $f$ over a $t$-length interval then it will have executed $t\cdot \gamma_{i,k} \cdot f$ units of computation. For the same interval, $\tau^\prime_i$, at frequency $1$, is executing $t\cdot \gamma^\prime_{i,k} \cdot 1 = t\cdot \gamma_{i,k} \cdot f$ units of computation. The amount of work executed per unit of time is exactly the same for every task of both systems. So if there exists one schedule without any deadline miss for one of the two systems, we can use the same one to schedule the other system. Thus, we can conclude that $\tau$ is feasible if and only if $\tau^\prime$ is feasible.
\end{proof}

\begin{theorem}
A necessary and sufficient condition for a sporadic task system $\tau$ (respecting \emph{sub-linear speedup ratio} and \emph{work-limited parallelism}) to be feasible on $m$ processors at frequency $f$ is given by:

\begin{equation}\label{eqn:sched-cond}
m \geq \sum_{i=1}^{n}
\left(k_{i}(f)+\frac{u_i-\gamma_{i,k_{i}(f)}\times f}{(\gamma_{i,k_{i}(f)+1}-
\gamma_{i,k_{i}(f)})\times f} \right)~.
\end{equation}
\end{theorem}

\begin{proof}
By Theorem~\ref{theo:tauprime}, we know that $\tau$ is feasible at frequency $f$ on $m$ processing cores if and only if $\tau^\prime$ is feasible at frequency $1$. In~\cite{CCG06b}, there is a necessary and sufficient feasibility condition for any sporadic task system (work-limited and sub-linear speedup ratio) for fixed frequency ($f = 1$). This result can be used to establish the schedulability of $\tau^\prime$.

So, using the result given by~\cite{CCG06b}, we know that $\tau^\prime$ is feasible if and only if this inequation holds:

\begin{equation}\label{eqn:sched-condPrevPaper}
m \geq \sum_{i=1}^{n}
\left(k^\prime_{i}+\frac{u_i-\gamma^\prime_{i,k^\prime_{i}}}{\gamma^\prime_{i,k^\prime_{i}+1}-
\gamma^\prime_{i,k^\prime_{i}}} \right)~,
\end{equation}

where $k^\prime_i$ denotes the value of $k_i$ (\emph{cf.} definition given by~(\ref{eq:smalNbrProc})) calculated for the system $\tau^\prime$ at frequency $1$. $\forall \: 1 \leq i \leq n :$

\begin{equation*}
k^\prime_i = k^\prime_i(1) = k_{i}(f)
\end{equation*}

We can now replace $k^\prime_i$ and $\gamma^\prime_{i,k}$ by their value in~(\ref{eqn:sched-condPrevPaper}):

\begin{align*}
& m \geq \sum_{i=1}^{n}
\left(k_{i}(f)+\frac{u_i-\gamma_{i,k_{i}(f)} \times f}{\gamma_{i,k_{i}(f)+1} \times f -
\gamma_{i,k_{i}(f)} \times f} \right) \\
\Leftrightarrow \: &
m \geq \sum_{i=1}^{n}
\left(k_{i}(f)+\frac{u_i-\gamma_{i,k_{i}(f)}\times f}{(\gamma_{i,k_{i}(f)+1}-
\gamma_{i,k_{i}(f)})\times f} \right)~, \\
\end{align*}

which corresponds exactly to~(\ref{eqn:sched-cond}). So $\tau^\prime$ is feasible if and only if~(\ref{eqn:sched-cond}) holds. Thus, by Theorem~\ref{theo:tauprime}, $\tau$ is feasible if and only if~(\ref{eqn:sched-cond}) holds.
\end{proof}

\begin{definition}[Minimum number of processor function for parallel tasks and system]

For any $\tau_i \in \tau$:
\begin{equation*}
M_{i}(f) \equals k_{i}(f)+\frac{u_i-\gamma_{i,k_{i}(f)}\times f}{(\gamma_{i,k_{i}(f)+1}-
\gamma_{i,k_{i}(f)})\times f}
\end{equation*}

Therefore, we can define the same notion system-wide:
\begin{align*}
M_{\tau}(f) &\equals \sum_{i=1}^{n} M_{i}(f) \\
		& = \sum_{i=1}^{n} \left(k_{i}(f)+\frac{u_i-\gamma_{i,k_{i}(f)}\times f}{(\gamma_{i,k_{i}(f)+1}-\gamma_{i,k_{i}(f)})\times f} \right)
\end{align*}
\end{definition}

Based on this definition, the feasibility criteria~(\ref{eqn:sched-cond}) becomes:
\begin{equation}\label{eqn:feasibility_short}
m \geq M_{\tau}(f)
\end{equation}

Notice that, for a fixed frequency $f$, the minimum number of processors necessary and sufficient to schedule the system is $ \lceil M_{\tau}(f) \rceil $.

In the following, we will show that in our model, the feasibility of the system is sustainable regarding the frequency i.e. increasing the value of the frequency maintains the feasibility of the system. For this, we will need the following theorem.

\begin{theorem}\label{thm:m_decreasing}
$M_{\tau}(f)$ is a monotonically decreasing function for $f > 0$.
\end{theorem}

\begin{proof}
We will first prove someting stronger:
\begin{align}
& \forall \tau_i \in \tau, \forall f_1, f_2 \in \mathbb{R}^{+}_{0} : \notag \\
& 0 < f_1 \leq f_2 \Rightarrow M_i(f_1) \geq M_i(f_2) \label{eq:decreasing_mi}
\end{align}

First, notice that $\forall \, \tau_i \in \tau$, $k_i(f)$ is a decreasing staircase function. Indeed, the value of $k_i(f)$ depends on the satisfaction of $\gamma_{i,k} \times f < u_i$. In this inequation, the greater $f$ is, the smaller $\gamma_{i,k}$ has to be to hold it and so does $k$ because $\Gamma_i$ is ordered by model assumption ($0 < \gamma_{i,1} < \gamma_{i,2} < \ldots < \gamma_{i,m}$).

In order to confirm the decrease of $M_i(f)$, we have to consider both cases:
\begin{compactitem}
\item When $k_i(f)$ remains constant between two frequencies.
\item When $k_i(f)$ jumps a step between two frequencies.
\end{compactitem}

The case when $\kappa = m$ is trivial : $f < \frac{u_i}{\gamma_{i,m}}$, $M_i(f) = m$ and task $\tau_i$ is not schedulable on $m$ processors at this frequency. So before $f = \frac{u_i}{\gamma_{i,m}}$, $M_i(f)$ is constant and therefore monotonically decreases.

First case to consider is when $k_i(f) = \kappa$ is fixed ($\kappa \in \lbrace 0, 1, \ldots, m-1 \rbrace$). By~(\ref{eqn:inverse-smalNbrProc}), we have that for $f$ in $\left[ \frac{u_i}{\gamma_{i,\kappa+1}}, \frac{u_i}{\gamma_{i,\kappa}} \right) $, $k_i(f) = \kappa$ remains constant and
\begin{align*}
M_i(f) &= \kappa + \frac{u_i - \gamma_{i,\kappa} \times f}{(\gamma_{i,\kappa+1} - \gamma_{i,\kappa}) \times f} \\
		  &= \kappa + \frac{u_i}{(\gamma_{i,\kappa+1} - \gamma_{i,\kappa}) \times f} - \frac{\gamma_{i,k}}{\gamma_{i,\kappa+1} - \gamma_{i,\kappa}}
\end{align*}

\noindent decreases as a multiplicative inverse function (terms which don't depends on $f$ are fixed in this interval).\\

Now consider the case when there is a variation in the value of $k_i(f)$. This occurs only when $f = \frac{u_i}{\gamma_{i,k}}$ for $k = m, m-1, \ldots, 1$. At this exact value of the frequency, the value of $k_i(f)$ jumps from $k$ to $k-1$. We will prove that even in those cases, the function $M_i(f)$ still decreases.

We have to prove the following:
\begin{equation*}
\frac{u_i}{\gamma_{i,k+1}} \leq f^\prime < \frac{u_i}{\gamma_{i,k}} = f \quad \Rightarrow \quad M_i(f') > M_i(f)
\end{equation*}

Let us compute their values individually:
\begin{align*}
& \forall \, 1 \leq k \leq m : \\
& k_i(f) = k_i \left( \frac{u_i}{\gamma_{i,k}} \right) = k - 1 \\
\Rightarrow \; &
M_i(f) = M_i \left( \frac{u_i}{\gamma_{i,k}} \right) \\
		  & \qquad \; \, = k-1 + \frac{u_i - \gamma_{i,k-1} \times \frac{u_i}{\gamma_{i,k}} }{(\gamma_{i,k} - \gamma_{i,k-1}) \times \frac{u_i}{\gamma_{i,k}}} \\
		  & \qquad \; \, = k - 1 + 1 \\
		  & \qquad \; \, = k \\
& k_i(f^\prime) = k \quad \text{because} \quad f^\prime \in k_i^{-1}(k) \\
\Rightarrow \; &
M_i(f^\prime) = k + \frac{u_i - \gamma_{i,k} \times f^\prime}{(\gamma_{i,k+1} - \gamma_{i,k}) \times f^\prime}
\end{align*}

We know $u_i - \gamma_{i,k} \times f^\prime > 0$ because $f^\prime < \frac{u_i}{\gamma_{i,k}}$, $\gamma_{i,k+1} > \gamma_{i,k}$ is true by model assumption and $f^\prime > 0$. Thus we have, with $\varepsilon > 0$:
\begin{align*}
& M_i(f') = M_i(f) + \varepsilon \\
\Rightarrow \; & M_i(f') > M_i(f)
\end{align*}

We have now proved~(\ref{eq:decreasing_mi}) for every possible value of $k_i(f)$. Thus:
\begin{align*}
& f_1, f_2 \in \mathbb{R}^+_0 : 0 < f_1 \leq f_2 \\
&\Rightarrow M_i(f_1) \geq M_i(f_2) \qquad \forall \tau_i \in \tau \\
&\Rightarrow \sum_{i=1}^{n} M_{i}(f_1) \geq \sum_{i=1}^{n} M_{i}(f_2) \\
&\Rightarrow M_{\tau}(f_1) \geq M_{\tau}(f_2)
\end{align*}
\end{proof}

This theorem directly implies the following property.

\begin{property}\label{prop:sustainable_frequency}
The feasibility of the system is sustainable regarding the frequency\footnote{i.e., increasing the frequency preserves the system schedulability.}.
\end{property}

\begin{proof}
By~(\ref{eqn:feasibility_short}), $\tau$ is feasible on $m$ processors at frequency $f > 0$, if and only if $m \geq M_{\tau}(f)$.

If $\tau$ is feasible on $m$ processors at frequency $f > 0$,  then $\tau$ is feasible on $m$ processors at any greater frequency $f^\prime \geq f$ because, by Theorem~\ref{thm:m_decreasing}, the following holds:
\begin{align*}
& f \leq f^\prime \quad \text{and} \quad m \geq M_{\tau}(f) \\
\Rightarrow \; & m \geq M_{\tau}(f) \geq M_{\tau}(f^\prime) \\
\Rightarrow \; & m \geq M_{\tau}(f^\prime)~,
\end{align*}
which corresponds to the feasibility criteria at frequency $f^\prime$.
\end{proof}

\subsection{Minimum optimal frequency}\label{subs:minfreq}

Property~\ref{prop:sustainable_frequency} implies that there is a minimum frequency for the system to be feasible. Then, it would be interesting to have an algorithm to compute it for a particular task system $\tau$ and a maximum number of processors $m$. We will first derive a constraint on the frequency from the feasibility criteria. After that, we will use this constraint to design an algorithm that computes the optimal minimum frequency in $O(n^2 \log^2_2 m)$ time.

\begin{definition}[Minimum frequency notation]
\begin{equation*}\label{eqn:feasible-freq}
\Psi(\tau, m, \vec{\kappa}) \equals \frac{\sum_{i=1}^n \frac{u_i}{\gamma_{i,\kappa_i +1} - \gamma_{i,\kappa_i}}}{m-\sum^n_{i=1} \left(\kappa_i - \frac{\gamma_{i,\kappa_i}}{\gamma_{i,\kappa_i +1} - \gamma_{i,\kappa_i}} \right)}~,
\end{equation*}
where $\vec{\kappa} = (\kappa_1, \kappa_2, \ldots, \kappa_n)$.
\end{definition}

\begin{property}\label{prop:min_freq_constr}
If $\vec{k}(f) = (k_1(f), k_2(f), \ldots, k_n(f))$, then the following holds:
\begin{equation}
m \geq M_\tau(f) \Leftrightarrow f \geq \Psi(\tau, m, \vec{k}(f))
\end{equation}
\end{property}

\begin{proof}
Let us define a few more notations:
\begin{align*}
M^{\prime}_{\tau}(f) &= \sum_{i=1}^{n} \left( k_{i}(f) - \frac{\gamma_{i,k_i(f)}}{\gamma_{i,k_i(f)+1} - \gamma_{i,k_i(f)}} \right) \\
M^{\prime\prime}_{\tau}(f) &= \sum_{i=1}^{n} \left( \frac{u_i}{\gamma_{i,k_i(f)+1} - \gamma_{i,k_i(f)}} \right) \\
\end{align*}
A few things to notice:
\begin{compactitem}
\item we have $M^{\prime\prime}_{\tau}(f) > 0$ because $u_i > 0$ (tasks aren't trivial) and $\gamma_{i,k} < \gamma_{i,k+1} \, \forall \, k\in\lbrace 0, 1, \ldots, m \rbrace$ (sub-linear speedup ratio).
\item $M_{\tau}(f) = M^{\prime}_{\tau}(f) + \frac{M^{\prime\prime}_{\tau}(f)}{f} \quad$ (with $f > 0$).
\item the last two items implies that $M_\tau(f) > M^\prime_\tau(f)$
\end{compactitem}

We have:
\begin{align}
& m \geq M_\tau(f) > M^\prime_\tau(f) \notag \\
\Rightarrow \; & m > M^\prime_\tau(f) \notag \\
\Leftrightarrow \; & m - M^\prime_\tau(f) > 0 \label{eqn:m_mgt0}
\end{align}

And:
\begin{align*}
& m \geq M_\tau(f) \\
\Leftrightarrow \; & m \geq M^{\prime}_{\tau}(f) + \frac{M^{\prime\prime}_{\tau}(f)}{f} \\
\Leftrightarrow \; & \underbrace{m - M^{\prime}_{\tau}(f)}_{> 0 \text{~by~(\ref{eqn:m_mgt0})}} \geq \frac{M^{\prime\prime}_{\tau}(f)}{\underbrace{f}_{> 0}} \\
\Leftrightarrow \; & f \geq \frac{M^{\prime\prime}_{\tau}(f)}{m - M^{\prime}_{\tau}(f)} = \Psi(\tau, m, \vec{k}(f))
\end{align*}
\end{proof}

Let us denote by $f_{MIN}$ the optimal minimum frequency such that the system $\tau$ is feasible on $m$ processors. By Property~\ref{prop:min_freq_constr}, $f_{MIN}$ is the smallest real positive number $f$ such that
\begin{equation}\label{def:freqminopt}
f \geq \Psi(\tau, m, \vec{k}(f))~.
\end{equation}

Consider fixing each $k_i(f)$ term such that they are equal to $k_i(f_{MIN})$. From there, it would be easy to calculate $f_{MIN}$ with the function $\Psi$ (in $O(n)$ time).

The first thing the algorithm will do is then searching those values (denoted by $\bar{\kappa}_1, \bar{\kappa}_2, \ldots, \bar{\kappa}_n$ such that $\bar{\kappa}_i = k_i(f_{MIN}) \; \forall \tau_i \in \tau$) and then compute the value of $f_{MIN}$ with the expression $\Psi(\tau, m, \bar{\kappa} = (\bar{\kappa}_1, \bar{\kappa}_2, \ldots \bar{\kappa}_n))$. The algorithm will be presented in the next section.

\subsection{Algorithm Description}\label{subsec:algo-desc}

\begin{algorithm}
\caption{$\mathsf{feasible}(\tau, m, f)$}\label{alg:feasibility}
$\mathsf{sum} \leftarrow 0$ \\
\For{$\tau_i \in \tau$}
{
		$\kappa_i \leftarrow k_i(f)$ \\
		$\mathsf{sum} \leftarrow \mathsf{sum} + \kappa_i + \frac{u_i - \gamma_{i,\kappa_i} \times f}{(\gamma_{i,\kappa_i+1} - \gamma_{i,\kappa_i}) \times f} $ \\
}
\textbf{return} $m \geq \mathsf{sum}$
\end{algorithm}

\begin{algorithm}
\caption{$\mathsf{minimumOptimalFrequency}(\tau, m)$}\label{alg:minfreq}
\For{$i \in \lbrace 1, 2, \ldots, n \rbrace$}
{
	\If{$\mathsf{feasible}(\tau, m, \frac{u_i}{\gamma_{i,m}})$}
	{
		$\bar{\kappa}_i \leftarrow m-1$\\
	}
	\Else
	{
		$\bar{\kappa}_i \leftarrow \min_{\kappa = 0}^{m-1} \lbrace \kappa \mid \mathsf{\bf{not}} \; \mathsf{feasible}(\tau, m, \frac{u_i}{\gamma_{i,\kappa+1}}) \rbrace $
	}
}
$\bar{\kappa} \equals \left( \bar{\kappa}_1, \bar{\kappa}_2, \ldots, \bar{\kappa}_n \right) $ \\
$f_{MIN} \leftarrow \Psi(\tau, m, \bar{\kappa})$ \\
\textbf{return} $f_{MIN}$
\end{algorithm}

We have designed an algorithm to determine the optimal minimum frequency (see Algorithm~\ref{alg:minfreq}). The algorithm essentially systematically searches for the minimum frequency that that satisfies the constraints of~\ref{eqn:sched-cond} by calling the feasibility test function (Algorithm~\ref{alg:feasibility}). For each value $\kappa$ that we want to test, we determine from~(\ref{eqn:inverse-smalNbrProc}) the minimum frequency $f$ such that $\tau_i$ requires $\kappa+1$ processors (i.e, $f = k^{-1}_i(\kappa).f_1 = \frac{u_i}{\gamma_{i,\kappa+1}}$).  The value of $f$ can be determined in $O(1)$ time.

In the feasibility test, we determine the value of $k_i(f)$ from frequency $f$ according to~(\ref{eq:smalNbrProc}), which can be obtained in $O(\log_2 m)$ time by binary search over $m$ values. Thus, to calculate $k_i(f)$ for all $\tau_i \in \tau$ and sum every $M_i(f)$ terms, the total time complexity of the feasibility test is $O(n \log_2 m)$.

In the main algorithm aimed at calculating $f_{MIN}$, the value of $\bar{\kappa}_i$ can also be found by binary search and thus takes $O(\log_2 m)$ time to be computed. This is made possible by the sustainability of the system regarding the frequency (proofed by Property~\ref{prop:sustainable_frequency}). Indeed, if $\tau$ is feasible on $m$ processors with $\kappa_i$ ($f = \frac{u_i}{\gamma_{i,\kappa_i+1}}$), then it's also feasible with $\kappa_i-1$ ($f = \frac{u_i}{\gamma_{i,\kappa_i}} > \frac{u_i}{\gamma_{i,\kappa_i+1}}$).

In order to calculate the complete vector $\bar{\kappa}$, there will be $O(n \log_2 m)$ calls to the feasibility test. Since computing $\Psi$ is linear-time when the vector $\bar{\kappa}$ is already stored in memory, the total time complexity to determine the optimal feasible frequency is $O(n^2 \log_2^2 m)$.
In order to determine the optimal combination of frequency and number of processors, we simply iterate over all possible number of active processors $\ell = 1, 2, \ldots, m$ executing Algorithm~\ref{alg:minfreq} with inputs $\tau$ and $\ell$.  We return the combination that results in the minimum overall power-dissipation rate.  Thus, the overall complexity to find the optimal combination is $O(m n^2 \log_2^2 m)$.

\subsection{An Example}\label{subsec:algo-example}

Let us use the same example system than previously introduced in Section~\ref{subsec:scheduler}. Consider $\tau=\{\tau_1,\tau_2 \}$ to be scheduled on $m = 3$ identical processors. Tasks are defined as follow : $\tau_1=(6,4,\Gamma_1)$ with $\Gamma_1=(1.0, 1.5, 2.0)$ and $\tau_2=(3,4,\Gamma_2)$ with $\Gamma_2=(1.0, 1.2, 1.3)$. The vector $\bar{\kappa}$ corresponding to this configuration computed by the algorithm is equal to $(\bar{\kappa}_1 = 2, \bar{\kappa}_2 = 0)$. This implies that the optimal minimum frequency for this system to be feasible on $3$ processors is equal to $f_{MIN} = \Psi(\tau, 3, \langle2,0\rangle) = 0.9375$. We can see that if we call the feasibility test function for any frequency greater or equal than $0.9375$, it will return $\mathsf{True}$; it will return $\mathsf{False}$ for any lower value.

\subsection{Proof of Correctness}\label{subsec:algo-correct}

The efficiency and correctness of the above algorithm depends upon the theorem presented in Sections~\ref{subs:feasibility} and~\ref{subs:minfreq}. Furthermore, the algorithm is correct if the value $\Psi$ computed using the previously calculated vector $\bar{\kappa}$ is equal to the minimum optimal frequency as defined by~(\ref{def:freqminopt}). That will be the goal of our last theorem.

\begin{theorem}
\begin{equation}
f_{MIN} = \Psi(\tau, m, \bar{\kappa})
\end{equation}
\end{theorem}

\begin{proof}
We will need an auxiliary notion:
\begin{equation*}
M_{i}(f, \kappa) \equals \kappa + \frac{u_i - \gamma_{i, \kappa} \times f}{(\gamma_{i,\kappa+1} - \gamma_{i,\kappa}) \times f} \quad \kappa \in \lbrace 0, 1, \ldots, m-1 \rbrace
\end{equation*}

Notice the following:
\begin{equation*}
M_{i}(f) = M_{i}(f, k_{i}(f)) \quad \quad \forall \; \tau_i \in \tau
\end{equation*}

By definition,
\begin{align*}
f_{MIN} &= \min \lbrace f \in \mathbb{R}^+_0 \mid m \geq M_\tau(f) \rbrace \\
&= \min \lbrace f \in \mathbb{R}^+_0 \mid f \geq \Psi(\tau, m, \vec{k}(f)) \rbrace \quad \text{by Property~\ref{prop:min_freq_constr}} \\
&= \Psi(\tau, m, \vec{k}(f_{MIN}))\,.
\end{align*}
This is equivalent to
\begin{equation*}
m = M_\tau(f_{MIN}) \quad \text{by Property~\ref{prop:min_freq_constr}.}
\end{equation*}

We will prove the following:
\begin{align*}
\forall \; 1 \leq i \leq n: M_{i}(f_{MIN}, \bar{\kappa}_i) &= M_{i}(f_{MIN}, \vec{k}_i(f_{MIN})) \\
&= M_{i}(f_{MIN}) \\
\end{align*}
This would imply:
\begin{align*}
& \sum_{i=1}^{n} M_i(f_{MIN}, \bar{\kappa}_i) = M_{\tau}(f_{MIN}) = m \\
& \Leftrightarrow \Psi(\tau, m, \bar{\kappa}) = \Psi(\tau, m, \vec{k}(f_{MIN})) = f_{MIN}
\end{align*}

Notice that when $k_i(f_{MIN}) = \bar{\kappa}_i$, then we have $M_i(f_{MIN}) = M_i(f_{MIN},\bar{\kappa}_i)$.

We will have four cases of possible value of $\bar{\kappa}_i$ to investigate:
\begin{compactitem}
\item $\bar{\kappa}_i = m-1$, the basic case of the algorithm,
\item $\bar{\kappa}_i = 0$,
\item $\bar{\kappa}_i > 0$, when $f_{MIN} \neq \frac{u_i}{\gamma_{i,\bar{\kappa}_i}}$,
\item $\bar{\kappa}_i > 0$, when $f_{MIN} = \frac{u_i}{\gamma_{i,\bar{\kappa}_i}}$.
\end{compactitem}

Basic case, $\bar{\kappa}_i = m-1$:
\begin{align*}
\mathsf{feasible}(\tau, m, \frac{u_i}{\gamma_{i,m}}) &\Rightarrow f_{MIN} \leq \frac{u_i}{\gamma_{i,m}} < \frac{u_i}{\gamma_{i,m-1}} < \ldots < \frac{u_i}{\gamma_{i,1}} \\
&\Rightarrow k_i(f_{MIN}) \geq k(\frac{u_i}{\gamma_{i,m}}) = m-1\,,
\end{align*}
but for the system to be feasible, we must have $k_i(f_{MIN}) < m$, so:
\begin{equation*}
\Rightarrow m-1 \leq k_i(f_{MIN}) < m \Rightarrow k_i(f_{MIN}) = m-1
\end{equation*}

Complex case, $\bar{\kappa}_i = 0$:
\begin{align*}
\neg \mathsf{feasible}(\tau, m, \frac{u_i}{\gamma_{i,1}}) &\Rightarrow f_{MIN} > \frac{u_i}{\gamma_{i,1}} > \frac{u_i}{\gamma_{i,2}} > \ldots > \frac{u_i}{\gamma_{i,m}} \\
&\Rightarrow \frac{u_i}{\gamma_{i,1}} < f_{MIN} < \infty \\
&\Rightarrow f_{MIN} \in k^{-1}_{i}(0) \\
&\Rightarrow k_i(f_{MIN}) = 0
\end{align*}

Complex case, $\bar{\kappa}_i > 0$:
\begin{align*}
& \neg \mathsf{feasible}(\tau, m, \frac{u_i}{\gamma_{i,\bar{\kappa}_i+1}}) \; \wedge \; \mathsf{feasible}(\tau, m, \frac{u_i}{\gamma_{i,\bar{\kappa}_i}}) \\
& \Rightarrow \frac{u_i}{\gamma_{i,\bar{\kappa}_i+1}} < f_{MIN} \leq \frac{u_i}{\gamma_{i,\bar{\kappa}_i}}
\end{align*}

Case $f_{MIN} \neq \frac{u_i}{\gamma_{i,\bar{\kappa}_i}}$:
\begin{align*}
&\Rightarrow \frac{u_i}{\gamma_{i,\bar{\kappa}_i+1}} < f_{MIN} < \frac{u_i}{\gamma_{i,\bar{\kappa}_i}} \\
&\Rightarrow f_{MIN} \in \; ]\frac{u_i}{\gamma_{i,\bar{\kappa}_i+1}}, \frac{u_i}{\gamma_{i,\bar{\kappa}_i}}[ \; \subset k^{-1}_i(\bar{\kappa}_i) \\
&\Rightarrow k_i(f_{MIN}) = \bar{\kappa}_i
\end{align*}

Case $f_{MIN} = \frac{u_i}{\gamma_{i,\bar{\kappa}_i}}$:
\begin{align*}
& k_i(f_{MIN}) = \bar{\kappa}_i - 1 \\
&\Rightarrow M_i(f_{MIN}) = M_i(f_{MIN}, k_i(f_{MIN})) \\
& \qquad \qquad \qquad \, = M_i(\frac{u_i}{\gamma_{i,\bar{\kappa}_i}}, \bar{\kappa}_i-1) \\
& \qquad \qquad \qquad \, = \bar{\kappa}_i-1 + \frac{u_i - \gamma_{i,\bar{\kappa}_i-1} \times \frac{u_i}{\gamma_{i,\bar{\kappa}_i}}}{(\gamma_{i,\bar{\kappa}_i} - \gamma_{\bar{\kappa}_i-1}) \times \frac{u_i}{\gamma_{i,\bar{\kappa}_i}}} \\
& \qquad \qquad \qquad \, = \bar{\kappa}_i - 1 + 1 \\
& \qquad \qquad \qquad \, = \bar{\kappa}_i \\
& M_i(f_{MIN}, \bar{\kappa}_i) = \bar{\kappa}_i + \frac{u_i - \gamma_{i,\bar{\kappa}_i} \times \frac{u_i}{\gamma_{i,\bar{\kappa}_i}}}{(\gamma_{i,\bar{\kappa}_i+1} - \gamma_{\bar{\kappa}_i}) \times \frac{u_i}{\gamma_{i,\bar{\kappa}_i}}} \\
& \qquad \qquad \qquad = \bar{\kappa}_i
\end{align*}
\end{proof}

\section{Experimental Evaluation \& Simulation}\label{sec:evaluation}

In order to obtain realistic predictions regarding the effect of parallelism upon power consumption, we have evaluated our algorithm upon an actual hardware testbed.  In this section, we describe and discuss the high-level overview of the methodology employed in our evaluation, the low-level details involved in our evaluation methodology, and the results obtained from our experiments.

\subsection{Methodology Overview}\label{subsec:eval-overview}

Realistic predictions of the energy behavior of a real-time parallel system using our frequency-selection algorithm requires a hard-real-time parallel application to execute upon an instrumented multicore hardware testbed.  In the Compositional and Parallel Real-Time Systems (CoPaRTS) laboratory at Wayne State University, we have developed a power/thermal-aware testbed infrastructure to obtain accurate power and temperature readings.  Thus, we may obtain realistic hardware power measurements for any application executing on our testbed.

Regarding the hard-real-time parallel application, we are unfortunately not aware of any such available application that matches the malleable job model used in this paper\footnote{In fact, we are also unaware of any commercially or freely-available application for any of the other hard-real-time parallel job models.}.  However, given the continuous march of the real-time and embedded computing domains towards increasingly parallel architectures, we fully expect that such applications will be developed in the near future.  Thus, it behooves us to obtain as close to realistic as possible parameters for such future parallel real-time applications.   We have developed a methodology with this goal in mind.  Below is a high-level overview of the steps of our design methodology.  The details for each step are in the next subsection.

\begin{enumerate}
    \item\label{step:hardware} {\bf Modify Testbed}: We have modified a multicore platform to obtain accurate instantaneous CPU power readings.  Furthermore, our hardware testbed has the ability to run at a discrete set of frequencies and turn off individual cores.  Thus, our platform can approximately implement the frequencies determined from the frequency/processor-selection algorithm (Section~\ref{sec:algo}).
    \item\label{step:speedup} {\bf Obtain Realistic Speedup Vectors}:  Since we do not possess a hard-real-time application with malleable parallel jobs, we have observed the execution behavior of two different non-real-time parallel benchmarks (an I/O-constrained and non-I/O-constrained application) over different processing frequencies and levels of parallelism.  Our observations are used to construct two realistic speedup vectors to use in our stimulation (Step~\ref{step:simulation}).
    \item\label{step:power-readings}  {\bf Obtain Realistic Power Rates}:  Using the same non-real-time parallel benchmarks, we also construct a matrix of power dissipation rates over a range of processing frequencies and number of active cores.  Again, our measurements are utilized in the next simulation step.
    \item\label{step:simulation} {\bf Power-Savings Simulation}:  After obtaining the speedup vectors and corresponding power dissipation rates, we evaluate our algorithm over randomly generated task systems.  Our frequency/processor-selection algorithm is compared against the power required by an optimal non-parallel real-time scheduling approach (e.g., Pfair~\cite{BCPV96}).
\end{enumerate}

\subsection{Methodology Details}\label{subsec:eval-details}

\subsubsection{Testbed}\label{subsubsec:testbed}
For our testbed platform, we use an Intel $\emph{i}7~950$ processor with eight cores  (four physical cores with each physical core having two ``soft'' cores -- i.e., hyperthreads). The processor supports 13 different frequency settings.  (The processor sets the frequency level and all cores execute at the global frequency).  We use a Linux 2.6.27 kernel with PREEMPT-RT patch as our operating system. In addition, we have developed kernel modules for individual core shutdown and for frequency modulation functionality.

The testbed requires a few hardware modifications to measure the actual CPU power usage.  Towards this goal, we connect four shunt resisters, in-series ($.05\Omega$ each), with the four-wire eATX power connector interfaces of the motherboard (each $12$V power line is shunted with $0.05\Omega$ resisters). We measure the current ($A$) drawn by the CPU using National Instrument's NI 9205 Data Acquisition unit. Then, we calculate the total instantaneous CPU power (as the sum of all the individual powers) through each eATX $+12$V motherboard connectors. We run the testbed under the 13 different supported frequencies and active number of cores settings and record the corresponding power dissipation rates for the system.  When the number of active cores is less than eight, there is a choice of which core to shutdown.  To address this choice, we consider all the possible shutdown scenarios for a given number of active cores and use the average of the power-rate of all the scenarios. For example, in our eight-core processor, we have seven different ways to shutdown a single core\footnote{We cannot shutdown the core $0$ (boot core).}.  We calculate the power consumption of the system for each individual case and the average power is recorded as our final power-rate measurement for the combination of the given frequency and number of active cores.

\subsubsection{Speedup Vectors and Power Functions}

From our testbed, we can generate both a speedup vector and power-dissipation-rate function for non-I/O-constrained (i.e., CPU-bound) and I/O-constrained (i.e., memory-bound) parallel applications.  In order to obtain these parameters, we use two parallel applications: a modified version of \emph{Jetbench}~\cite{jetbench_Sim} for an non-I/O-consrained application and a modified parallel version of the \emph{GNU Compiler Collection} (GCC)~\cite{GCC_Ref} for an I/O-constrained application. Jetbench is an Open Source OpenMP-based multicore benchmark application that emulates the jet engine performance from real jet engine parameters and thermodynamic equations presented in the NASA's EngineSim program.  For GCC, using the ``-j'' option for GNU Make~\cite{GNU_Make_Ref}, we concurrently compile a collection of source code files under variable number of active processor cores.

To obtain the speedup vectors for both Jetbench and GCC, we execute the applications upon different numbers of active cores, recording for each number of cores the response time for the application.  The speedup for $x$ number of cores is determined by the ratio between the response time on one core to the response time of the application running concurrently on $x$ cores.  Figure~\ref{fig:speedup-vecs} plots the speedup vector for the two applications.  Not surprisingly, Jetbench benefits more greatly from increasing number of processors due to the CPU-bounded nature and inherently parallelizable workload.

\begin{figure}
\begin{center}
    \includegraphics[scale=.35, viewport=10 260 575 650, clip]{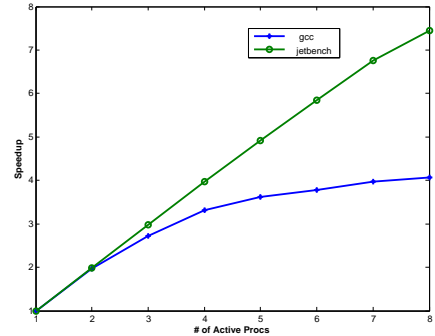}
\caption{Speedup Vectors for Jetbench and GCC}\label{fig:speedup-vecs}
\end{center}
\end{figure}

For determining the power-dissipation rates for both Jetbench and GCC, we execute these applications for all combinations of frequency and number of active cores and record both the power-dissipation rate and the speedup values for the application.  The power-dissipation rates are determined using the measurement hardware described above in Section~\ref{subsubsec:testbed}.  Each recorded value is an average of the power measured at a 1ms sampling intervals for the duration of the application.  Figure~\ref{fig:jetbench-power-map} plots the power-dissipation function for Jetbench; Figure~\ref{fig:gcc-power-map} plots the power-dissipation function for GCC.  Observe that the power-dissipation level for Jetbench are slightly higher in most cases than the levels for GCC; this is likely due to the fact that GCC idles the processor more often during I/O operations.


%
%

\begin{figure}
\begin{center}
    \includegraphics[scale=.50, viewport= 30 350 575 650, clip]{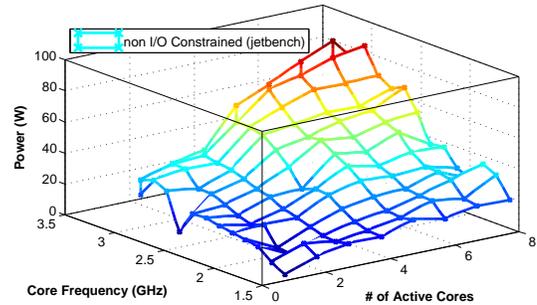}
\caption{Power Function Over Varying Number of Active Cores and Frequencies for Non-I/O-Constrained Parallel Application (Jetbench)}\label{fig:jetbench-power-map}
\end{center}
\end{figure}

\begin{figure}
\begin{center}
    \includegraphics[scale=.50, viewport= 10 350 575 650, clip]{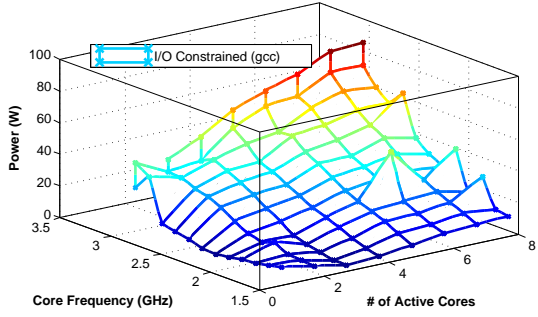}
\caption{Power Function Over Varying Number of Active Cores and Frequencies for I/O-Constrained Parallel Application (GCC)}\label{fig:gcc-power-map}
\end{center}
\end{figure}

\subsubsection{Power-Savings Simulation}

We randomly generate task systems using a variant of the \verb|UUnifast-Discard| algorithm by Davis and Burns~\cite{DB2009}.  In the \verb|UUnifast-Discard| algorithm, the user supplies a desired system-level utilization and number of tasks, and the algorithm returns a task system where each task has its task utilization randomly-generated from a uniform distribution each task utilization.  The difference between \verb|UUnifast-Discard| and the original \verb|UUnifast| algorithm from Bini and Buttazzoo~\cite{EG2004} is that \verb|UUnifast-Discard| generates task systems with system utilizations exceeding one, but task utilizations at most one.  These restrictions make \verb|UUnifast-Discard| appropriate for multiprocessor scheduling settings with non-parallel real-time jobs.  To extend the \verb|UUnifast-Discard|, we modified the algorithm to permit task utilizations to exceed one (i.e., a job is required to execute on more than one processor to complete by its deadline) and fix a single task at a given maximum utilization.  We call our extended algorithm \verb|UUnifast-Discard-Max|.  The utilization for each task generated by \verb|UUnifast-Discard-Max| (except for the task with fixed maximum utilization) is drawn from a uniform distribution.

Using the random-task generator, we generate task systems with a total of eight tasks.  The total system utilization is varied from $1.5$ to $8$ and the \verb|UUnifastDiscard_max| algorithm assigns a maximum utilization to the first task. We run our testbed with maximum utilization value $U_{\max}$ (i.e., $\max_{i=1}^n u_i$) equal to $.4$, $.8$, and $1.2$ in our simulations. Also, to match our testbed settings and the simulations, we select the number of CPUs from $1$ to $8$. The simulation runs for all the possible values of $1.5$ to $8$ utilization in $.1$ increments and number of available cores is varied from one through eight.  We run a variant of the Algorithm~\ref{alg:minfreq} that iterates through all frequencies and number of active core combinations, instead of using a binary search.  (Our power function does not exactly satisfy the non-decreasing property required for binary search to work). In each utilization point, we store the exact frequency returned by our algorithm.  For comparison, we determine the minimum frequency required for a optimal (non-parallel) scheduling algorithm to schedule the same task system.  This value can be obtained by solving the following for $f$: for any task system $\tau$, $U(\tau) \leq m f$.  Using these resulting frequencies, we obtain the optimal minimum frequency for the non-parallel and parallel settings.  We then use these frequencies to look up the power-dissipation rates for the respective application by using the functions displayed in Figures~\ref{fig:jetbench-power-map} and~\ref{fig:gcc-power-map}.  In the next subsection, we plot the power savings; i.e., we plot the power-dissipation level obtained from our algorithm minus the power-dissipation level required for the optimal non-parallel algorithm.  Each data point is the average power saving for 1000 different randomly-generated task systems.
\begin{figure}
\begin{center}
    \includegraphics[scale=.4, viewport= 40 300 575 600, clip]{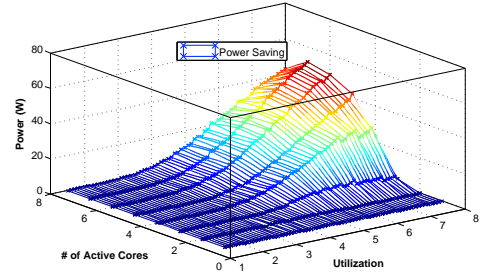}
\caption{Average power savings for non I/O Constrained workload (jetbench) when $U_{\max}=.4$}\label{fig:OPT_non_PRL_power_with_JB_U_max_point_4}
\end{center}
\end{figure}
\begin{figure}
\begin{center}
    \includegraphics[scale=.4, viewport= 55 300 575 600, clip]{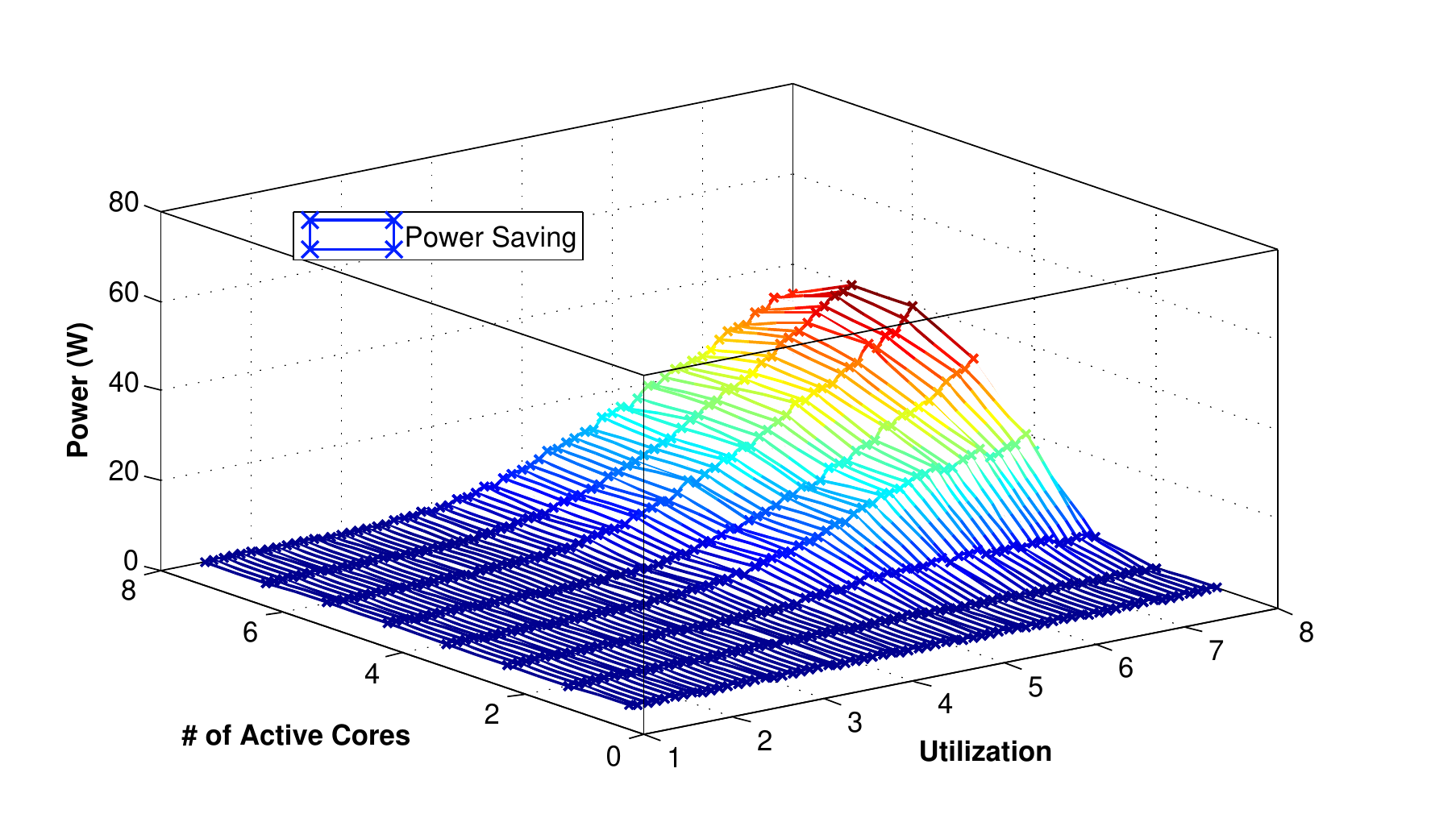}
\caption{Average power savings for non I/O Constrained workload (jetbench) when $U_{\max}=.8$}\label{fig:OPT_non_PRL_power_with_JB_U_max_point_8}
\end{center}
\end{figure}
\begin{figure}
\begin{center}
    \includegraphics[scale=.4, viewport= 30 300 575 600, clip]{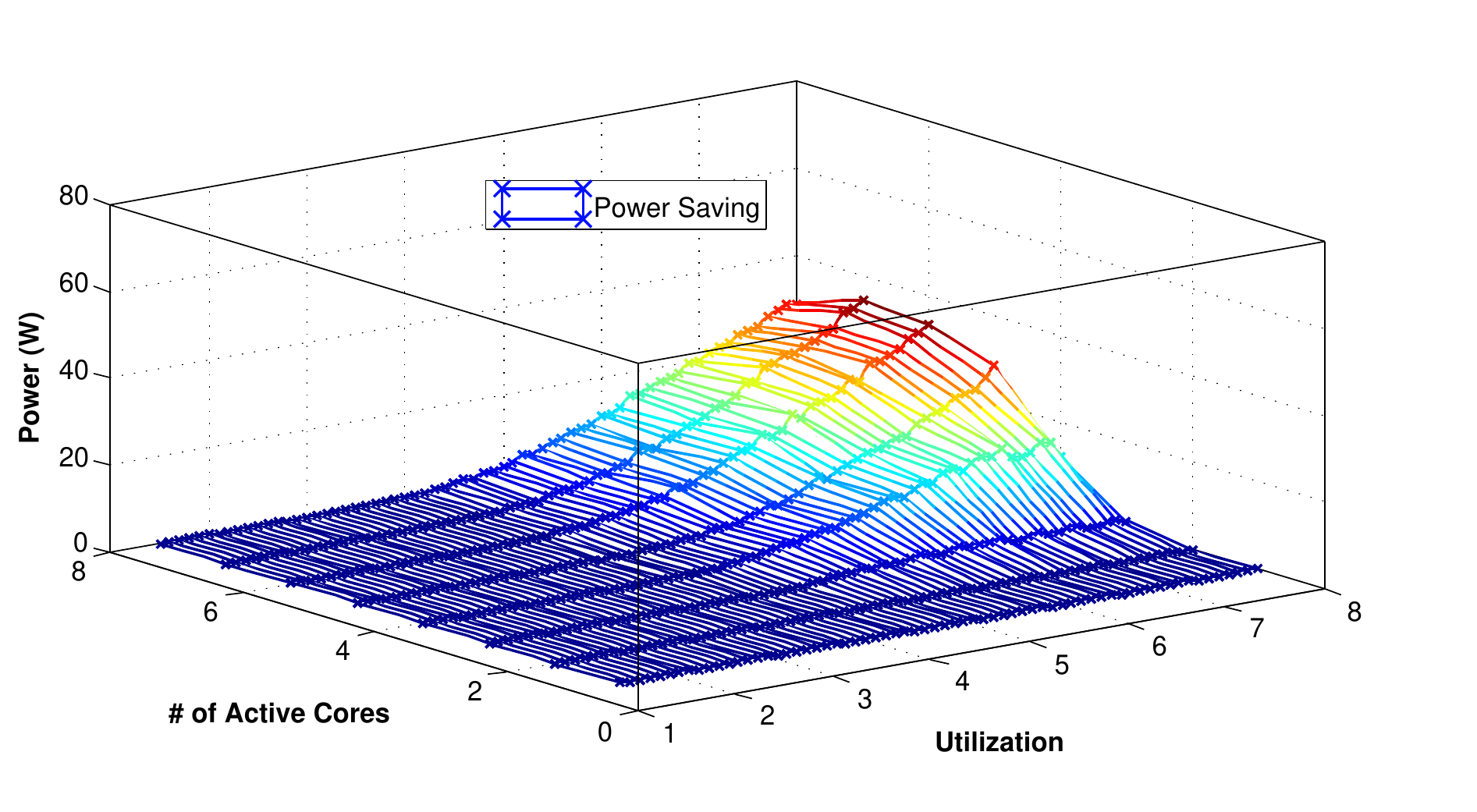}
\caption{Average power savings for non I/O Constrained workload (jetbench) when $U_{\max}=1.2$}\label{fig:OPT_non_PRL_power_with_JB_U_max_1_point_2}
\end{center}
\end{figure}

\begin{figure}
\begin{center}
    \includegraphics[scale=.4, viewport= 40 300 575 600, clip]{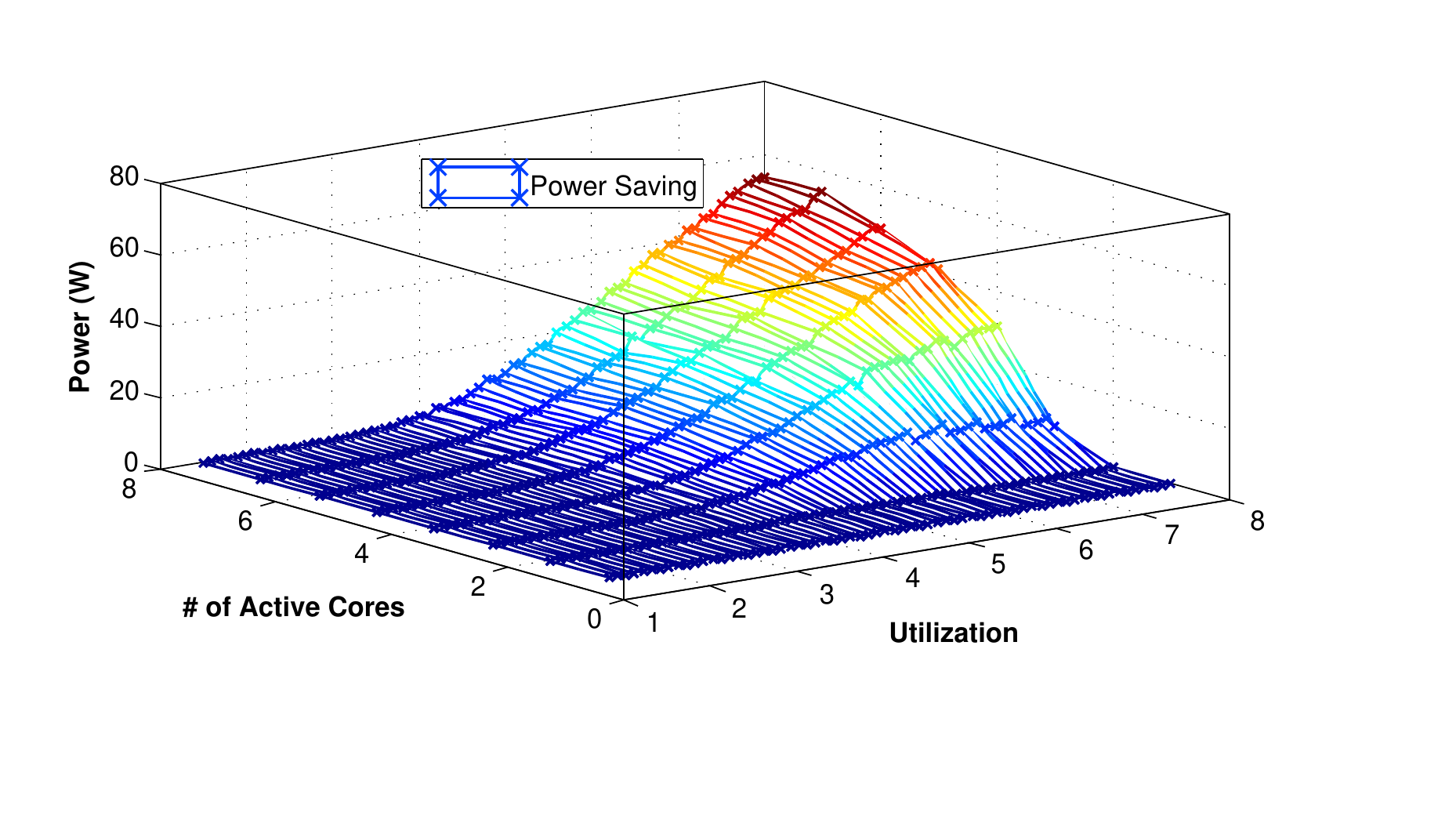}
\vspace{-.15in}\caption{Average power savings for I/O Constrained workload (gcc) when $U_{\max}=.4$}\label{fig:OPT_non_PRL_power_with_GCC_U_max_point_4}
\end{center}
\end{figure}


\begin{figure}
\begin{center}
    \includegraphics[scale=.4, viewport= 40 310 575 600, clip]{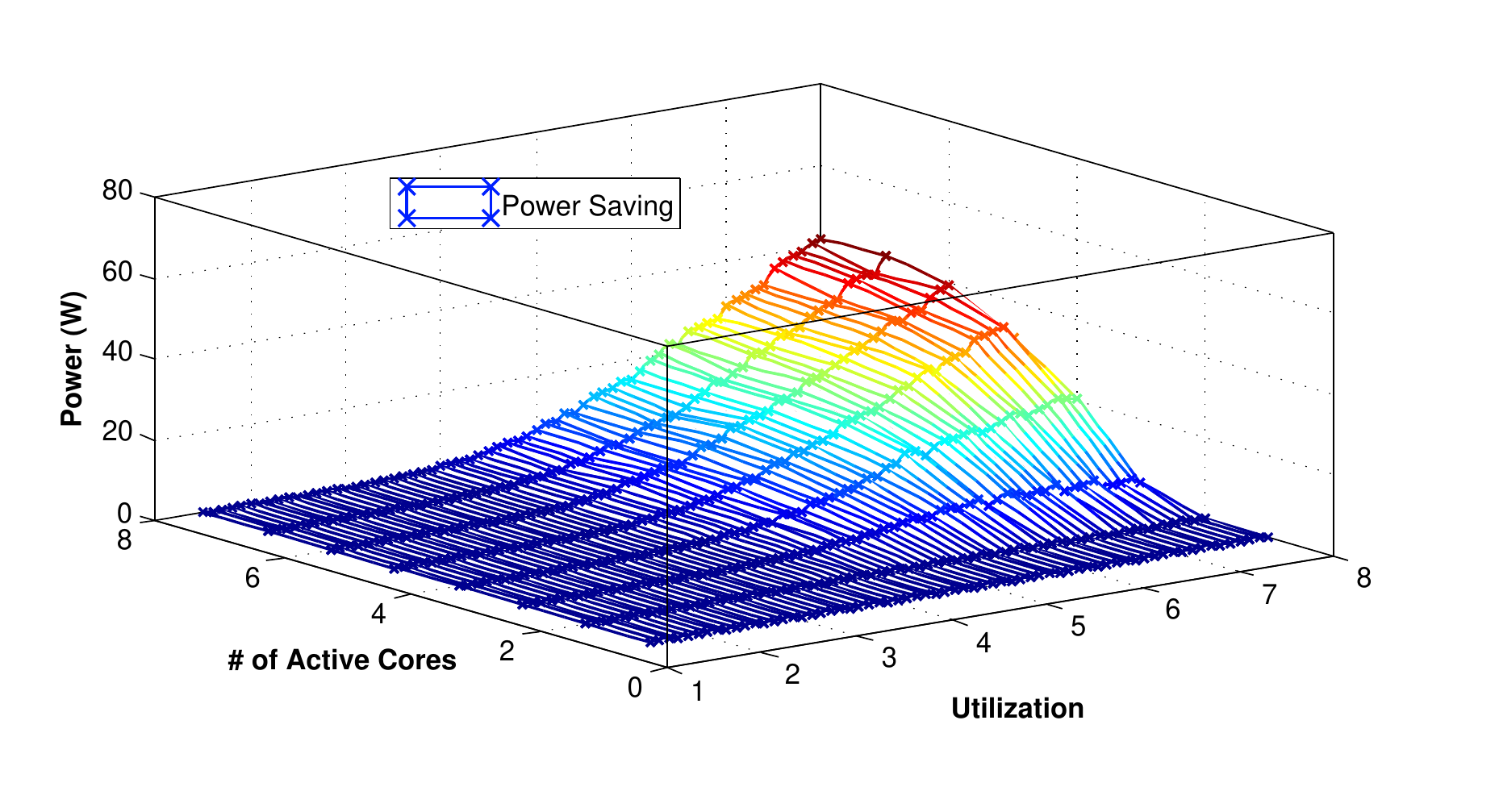}
\vspace{-.15in}\caption{Average power savings for I/O Constrained workload (gcc) when $U_{\max}=.8$}\label{fig:OPT_non_PRL_power_with_GCC_U_max_point_8}
\end{center}
\end{figure}


\begin{figure}
\begin{center}
    \includegraphics[scale=.4, viewport= 30 310 575 600, clip]{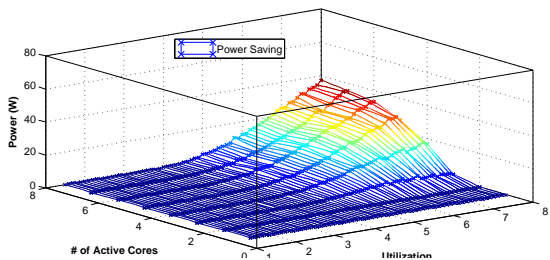}
\caption{Average power savings for I/O Constrained workload (gcc) when $U_{\max}=1.2$}\label{fig:OPT_non_PRL_power_with_GCC_U_max_1_point_2}
\end{center}
\end{figure}

\subsection{Results \& Discussion}\label{subsec:eval-results}

Figures~\ref{fig:OPT_non_PRL_power_with_JB_U_max_point_4},~\ref{fig:OPT_non_PRL_power_with_JB_U_max_point_8}, and~\ref{fig:OPT_non_PRL_power_with_JB_U_max_1_point_2} display the power savings obtained from simulating over the parallel/power parameters obtained from the Jetbench application.  Figures~\ref{fig:OPT_non_PRL_power_with_GCC_U_max_point_4},~\ref{fig:OPT_non_PRL_power_with_GCC_U_max_point_8}, and~\ref{fig:OPT_non_PRL_power_with_GCC_U_max_1_point_2} display the power savings for the GCC application.  The largest power savings is 60 watts (for GCC when $U_{\max} = .4$ and both the utilization and active cores equal eight) which is significant since from Figures~\ref{fig:jetbench-power-map} and~\ref{fig:gcc-power-map} the maximum power dissipation rate is around 80 watts.

From these plots, there are a few noticeable trends: 1) as $U_{\max}$ increases, the power savings decrease for both applications; the reason for this decrease is that larger utilization jobs require greater parallelization and thus more parallel overhead which reduces the power savings.  2) As the total utilization increases, the power savings increases (for active processors greater than two); in this case, the savings appears to be due to the fact that the power-dissipation rates are considerably higher at the highest core frequencies.  Thus, if our parallel algorithm can reduce the frequency over the non-parallel algorithm by a slight amount, there is significant power savings.  3) The power savings for both applications are similar; however, the I/O-constrained application, GCC, appears to have slightly higher power savings.  Again, the power-dissipation function for GCC may reward small frequency reductions slightly more than Jetbench's function.  Also, since we have a discrete set of frequencies, many of the different frequencies returned by Algorithm~\ref{alg:minfreq} will get mapped to the same core frequency reducing the differences for the two applications.

\section{Conclusions}\label{sec:conclusions}

In this paper, we explore the potential energy savings that could be obtained from exploiting parallelism present in a real-time application.  We consider the case of malleable Gang scheduled parallel jobs and design an optimal polynomial-time algorithm for determining the frequency to run each active core when we have the constraint of homogenous core frequencies.  Simulations with power data from an actual hardware testbed confirm the efficacy of our approach by providing significant power savings over the optimal non-parallel scheduling approach.  As real-time embedded systems are trending toward multicore architecture, our research suggests the potential in reducing the overall energy consumption of these devices by exploiting task-level parallelism.  In the future, we will extend our research to investigate power saving potential when the cores may execute at different frequencies and also incorporate thermal constraints into the problem.

\bibliographystyle{IEEEtran}
\bibliography{biblioEcrts13}
%

\end{document}